\documentclass[11pt]{article}
\usepackage[papersize={8.5in,11in}, margin=1in]{geometry}	% US letter with 1 inch margins
\usepackage{lmodern}										% smoothing out font edges

\usepackage[utf8]{inputenc}
\usepackage[T1]{fontenc}
\usepackage{graphicx}
\usepackage{soul}
\usepackage{amsmath,amsfonts,amsthm,amssymb}
\usepackage{thm-restate}
\usepackage{authblk}
\usepackage{xfrac}                                        % for nicer affiliation blocks
\usepackage{enumerate}
\usepackage[algoruled, vlined, linesnumbered]{algorithm2e}
\usepackage{setspace}

\graphicspath{{img/}}

\usepackage[dvipsnames]{xcolor}

\usepackage{tikz}
\usetikzlibrary{arrows.meta}
\usepackage[colorlinks=true, urlcolor=blue, linkcolor=blue, citecolor=ForestGreen]{hyperref}
\usepackage[noabbrev, nameinlink]{cleveref}					% Must be included after ›hyperref‹.
\newtheorem{theorem}{Theorem}%[section]
\newtheorem{corollary}[theorem]{Corollary}
\newtheorem{lemma}[theorem]{Lemma}

\crefname{modification}{modification}{modifications}

\theoremstyle{definition}

\newcommand*{\Otilde}{\widetilde{O}}

\newcommand*{\polylog}{\textsf{polylog}}

\newcommand*{\nwspace}{\hspace*{.1em}} 

\renewcommand{\leq}{\leqslant}
\renewcommand{\geq}{\geqslant}
\renewcommand{\le}{\leqslant}
\renewcommand{\ge}{\geqslant}
\renewcommand{\epsilon}{\varepsilon}

\let\oldsqrt\sqrt
\def\hksqrt{\mathpalette\DHLhksqrt}
\def\DHLhksqrt#1#2{\setbox0=\hbox{$#1\oldsqrt{#2\,}$}\dimen0=\ht0
   \advance\dimen0-0.2\ht0
   \setbox2=\hbox{\vrule height\ht0 depth -\dimen0}%
   {\box0\lower0.4pt\box2}}
\renewcommand\sqrt\hksqrt

\makeatletter
\newcommand{\removelatexerror}{\let\@latex@error\@gobble}
\makeatother

\providecommand{\ignore}[1]{} 

\title{Improved Approximate Distance Oracles:\\ Bypassing the Thorup-Zwick Bound in Dense Graphs}

\date{}

\author[1]{Davide Bilò}
\author[2]{Shiri Chechik}
\author[3]{Keerti Choudhary}
\author[4]{\\\vspace*{.5em} Sarel Cohen}
\author[5]{Tobias Friedrich}
\author[6]{Martin Schirneck}

\affil[1]{Department of Information Engineering, Computer Science and Mathematics, 
	
	University of L'Aquila

	\texttt{davide.bilo@univaq.it}
	\vspace*{.5em}
}
\affil[2]{Department of Computer Science, Tel Aviv University

	\texttt{shiri.chechik@gmail.com}
	\vspace*{.5em}
}
\affil[3]{Department of Computer Science and Engineering, Indian Institute of Technology Delhi,

	\texttt{keerti@iitd.ac.in}
	\vspace*{.5em}
}
\affil[4]{School of Computer Science, The Academic College of Tel Aviv-Yaffo

	\texttt{sarelco@mta.ac.il}
	\vspace*{.5em}
}
\affil[5]{Hasso Plattner Institute, University of Potsdam

	\texttt{tobias.friedrich@hpi.de}
	\vspace*{.5em}
}
\affil[6]{Faculty of Computer Science, University of Vienna

	\texttt{martin.schirneck@univie.ac.at}
}

\begin{document}
\maketitle

\begin{abstract}
Despite extensive research on distance oracles, there are still large gaps between the best constructions for spanners and distance oracles. Notably, there exist sparse spanners with a multiplicative stretch of $1+\epsilon$ plus some additive stretch.
A fundamental open problem is whether such a bound is achievable for distance oracles as well. Specifically, can we construct a distance oracle with multiplicative stretch better than 2,
along with some additive stretch, while maintaining subquadratic space complexity?
This question remains a crucial area of investigation, and finding a positive answer would be a significant step forward for distance oracles.
Indeed, such oracles have been constructed for sparse graphs. However, in the more general case of dense graphs, it is currently unknown whether such oracles exist.

In this paper, we contribute to the field by presenting the first distance oracles that achieve a multiplicative stretch of $1+\epsilon$ along with a small additive stretch while maintaining subquadratic space complexity. Our results represent an advancement particularly for constructing efficient distance oracles for dense graphs.
In addition, we present a whole family of oracles that, for any positive integer $k$, achieve a multiplicative stretch of $2k-1+\varepsilon$ using $o(n^{1+1/k})$ space.
\end{abstract}

\section{Introduction}
\label{sec:intro}

The increasing scale and complexity of modern networks, including social networks, communication networks, and the Internet itself, necessitate efficient algorithms for solving fundamental graph problems. Among these problems, distance approximation has attracted significant attention due to its wide applicability in network routing, traffic engineering, distributed computing, and numerous other domains. 
Distance oracles provide a powerful tool to estimate the shortest path distances between pairs of vertices in a network, enabling efficient computations even in the presence of massive graphs.
A distance oracle is considered to have a stretch of $(\alpha,\beta)$ (or is referred to as an $(\alpha,\beta)$-approximate distance oracle) if, for any pair of vertices $s$ and $t$, the value $\widehat{d}(s,t)$ returned by the oracle satisfies $d(s,t) \leq \widehat{d}(s,t) \leq \alpha \cdot d(s,t)+ \beta$.
Traditionally, the focus of designing distance oracles has been on minimizing both the space and the stretch of the oracle.

One way to construct an distance oracle is by employing an all-pairs shortest paths algorithm and storing the resulting distance matrix. By utilizing the pre-computed matrix, distance queries can be answered exactly in constant time. However, this approach has notable drawbacks: the storage space required can be excessively large (quadratic in the number of nodes), and computing all-pairs shortest paths can be time-consuming. To address these limitations, much of the research on distance oracles focuses on approximating distances instead.

Extensive research has been dedicated to developing approximate distance oracles. A seminal result in this field was achieved by Thorup and Zwick~\cite{ThorupZ05}. They introduced a construction method for approximate distance oracles that exhibit a stretch of $2k-1$. These oracles requires  $O(kn^{1+1/k})$ space and can be preprocessed in $O(mn^{1/k})$ time, where $k \geq 1$ is an integer parameter. 
See also subsequent research \cite{Chechik14,Chechik15,wulff2013approximate} for improvements on the query time and space requirements.

Thorup and Zwick additionally showed that distance oracles for undirected graphs must have $\Omega(n^2)$ space or a multiplicative stretch $\alpha \ge 3$ (provided that $\beta = 0$).
More generally, assuming the Erd\H{o}s girth conjecture, a distance oracle with $\alpha < 2k+1$ must have space $\Omega(n^{1+1/k})$.
There were numerous attempts to circumvent this limitation,
see, e.g.,~\cite{Abraham11AffineStretch,Agarwal14SpaceStretchTimeTradeoffDOs, AgarwalBrightenGodfrey13DOsStretchLessThan2,PatrascuRoditty14BeyondThorupZwick,Patrascu12NewInfinity} (a more detailed discussion can be found below).
However, all of these endeavors essentially shifted the focus towards specific settings or graph classes where the Thorup-Zwick bound does not apply. Specifically, none of these approaches considered dense graphs.

For example, Agarwal and Brighten Godfrey~\cite{Agarwal14SpaceStretchTimeTradeoffDOs,AgarwalBrightenGodfrey13DOsStretchLessThan2} investigated the possibility of constructing a distance oracle with a stretch less than 2, albeit at the expense of slower query times. They demonstrated that for any integer $t \geq 1$ and any $0 < c \le 1$, it is possible to devise a distance oracle of size $O(m + n^{2-c})$ that provides distances with a stretch of $(1 + \frac{1}{t})$. The query time for this oracle is $O((n^c\mu)^{t})$, where $\mu = \frac{2m}{n}$ is the average degree of the graph.
Furthermore, they also showed that the query time can be reduced to $O((n^c+\mu)^{2t-1})$ at the cost of a small additive stretch  $\frac{2t-1}{t}W$, where $W$ denotes the maximum edge weight in the graph.
The constructions in~\cite{Agarwal14SpaceStretchTimeTradeoffDOs,AgarwalBrightenGodfrey13DOsStretchLessThan2} indeed achieves a multiplicative stretch of less than 2, surpassing the limitations of the Thorup-Zwick bound.
However, their approach comes with two drawbacks.
Firstly, their method is applicable only to sparse graphs, as the space required for the distance oracle exceeds the number of edges in the graph. This makes it less suitable for dense graph scenarios.
Secondly, while they achieve a better stretch, the query time of their distance oracle is slower compared to other approaches.

To the best of our knowledge, no prior distance oracles have been constructed that simultaneously have subquadratic space and a multiplicative stretch better than two, even when allowing a small \emph{additive} stretch.
This poses the natural and significant open problem whether such oracles are possible.
In this paper, we provide an affirmative answer to this question.
For the most general setting of arbitrarily dense graphs, we give the first distance oracles with subquadratic space with a multiplicative stretch arbitrarily close to 1, by introducing an additive stretch.
Moreover, we give a whole family of distance oracles that, for any positive integer $k$, achieve a multiplicative stretch of $2k-1+\varepsilon$ using $o(n^{1+1/k})$ space.
Our results are summarized by the following theorems.

\pagebreak

\begin{restatable}{theorem}{distanceoracle}
\label{thm:distance_oracle}
	Let $G$ be an undirected graph with $n$ vertices and maximum edge weight $W$.
	For every integer $K = \omega(\log n) \cap O(\sqrt{n})$ and every $\varepsilon > 0$,
	there is a path-reporting distance oracle for $G$ with stretch $(1{+}\varepsilon, 2W)$,
	space $O(\frac{n^2}{K} \log n)$.
	The query time is $O(K^{\lceil 1/\varepsilon \rceil})$
	for the distance and an additional $O(1)$ per edge if the path is reported.
	The oracle is preprocessed in APSP time.
\end{restatable}

The Thorup-Zwick bound~\cite{ThorupZ05} can be adapted with
a simple information theoretic argument in bipartite graphs.
This way, one can unconditionally rule out the existence of 
distance oracles that distinguish between distances $1$ and $3$ in subquadratic space
(for neighboring vertices).
That means one cannot reduce the additive stretch in \Cref{thm:distance_oracle} to $1$
even in unweighted graphs.

By setting $\varepsilon = 1/t$ and $K^t = n^c$, we get the following stretch-space-time trade-off.

\begin{corollary}
\label{cor:distance_oracle}
	For every positive integer $t$ and real number $0 < c \le \frac{t}{2}$,
	there is a distance oracle for undirected graphs with stretch $(1{+}\frac{1}{t}, 2W)$,
	space $\Otilde(n^{2-\frac{c}{t}})$, and query time $O(n^c)$.
\end{corollary}

\begin{restatable}{theorem}{hierarchydo}
\label{thm:hierarchy_DO}
	Let $G$ be an undirected graph with $n$ vertices and maximum edge weight $W$.
	For positive integers $k$ and $K$, where $K = \omega(\log n) \cap O(n^{1/(2k+1)})$, and every $\varepsilon > 0$,
	there is a distance oracle for $G$ with stretch $(2k{-}1{+}\varepsilon, 4kW)$,
	space $O((\frac{n}{K})^{1+1/k} \, \log^{1+1/k} n)$, and query time $O(K^{2\lceil 4k/\varepsilon \rceil})$.
	The oracle is preprocessed in APSP time.
\end{restatable}

For $k=1$, the oracle in \Cref{thm:hierarchy_DO} has the same multiplicative stretch of $1+\varepsilon$
as the one in \Cref{thm:distance_oracle} and a better space of $\Otilde(n^2/K^2)$,
but the additive stretch of $4W$ is larger and so is the query time.
We obtain the following corollary for general $k$, $\varepsilon = 1/t$, and $K = n^{c/8kt}$.

\begin{corollary}
\label{cor:2k-1+eps-oracle}
	For all positive integers $k$ and $t$ and real number $0 < c \le (4{-}\frac{4}{2k+1}) \, t$,
	there is a distance oracle with stretch $(2k{-}1{+}\frac{1}{t},4kW)$,
	space $\Otilde(n^{1+\frac{1}{k}(1-\frac{c}{8t})})$, and query time $O(n^{c})$.
\end{corollary}

\paragraph{Related Work.}

Thorup and Zwick~\cite{ThorupZ05} showed for any positive integer $k$
that, assuming Erd\H{o}s' girth conjecture, 
any distance oracle for undirected graphs with a multiplicative stretch less than $2k + 1$
must have space of $\Omega(n^{1+1/k})$ bits.\footnote{
	Thorup and Zwick~\cite{ThorupZ05} also showed unconditionally
	that distance oracles for \emph{directed} graphs with arbitrary
	finite stretch must take $\Omega(n^2)$ bits of space.
}
The lower bound only applies
to graphs that are sufficiently dense and to queries that involve pairs of neighboring vertices,
leading to several attempts to bypass it in different settings.

There is a line of work on improved distance oracles for sparse graphs.
Porat and Roditty~\cite{PoratR11} showed that 
that for unweighted graphs and any $\varepsilon>0$, one can construct a distance oracle 
with multiplicative stretch $1{+}\varepsilon$ 
using space $O(nm^{1-\frac{\varepsilon}{4+2\varepsilon}})$,
which is subquadratic for $m = o(n^{1+\frac{\varepsilon}{4+\varepsilon}})$.
The query time is $\widetilde O(m^{1-\frac{\varepsilon}{4+2\varepsilon}})$.
For weighted graphs with average degree $\mu$ and any positive integer $t$,
Agarwal and Brighten Godfrey~\cite{AgarwalBrightenGodfrey13DOsStretchLessThan2} gave an oracle, 
for every positive integer $\lambda = O(n)$,
with multiplicative stretch $1+ \frac{1}{t}$, space $O(\mu n + \frac{n^2}{\gamma})$,
and query time $O((\mu \gamma)^{2t-1})$.
Agarwal~\cite{Agarwal14SpaceStretchTimeTradeoffDOs}
subsequently reduced the query time to $O((\mu \gamma)^t)$.

A notion that is closely related to distance oracles is that of spanners.
For those, it was shown that for larger distances one can get better stretch. 
For instance, Parter~\cite{Parter14} introduced the concept of \emph{hybrid stretch}
by constructing a spanner with $O(k^2 n^{1+1/k})$ edges that has a multiplicative stretch $2k-1$ for neighboring pairs and $k$ for all others.
Other attempts to avoid the Thorup-Zwick bound was to consider certain graph classes.
Notably, Fredslund-Hansen, Mozes, and Wulff-Nilsen~\cite{fredslundhansen21TrulySubquadraticPlanar} obtained a subquadratic-space distance oracle
for exact distances in planar graphs.

The smallest multiplicative stretch of an Thorup-Zwick~\cite{ThorupZ05} oracle
with subquadratic space is~$3$ (the case $k = 2$).
P\v{a}tra\c{s}cu and Roditty~\cite{PatrascuRoditty14BeyondThorupZwick} 
were arguably the first to introduce an additive stretch to 
simultaneously reduce the space below quadratic and the multiplicative stretch below $3$
in general dense graphs.
They proposed an oracle 
for unweighted graphs with stretch $(2,1)$ using $O(n^{5/3})$ space
that can be constructed in time $O(mn^{2/3})$.
P\v{a}tra\c{s}cu and Roditty~\cite{PatrascuRoditty14BeyondThorupZwick}
also showed that distance oracles with multiplicative stretch $\alpha \le 2$
require $\Omega(n^2)$ space (for weighted graphs),
assuming conjecture on the hardness of set intersection queries.
P\v{a}tra\c{s}cu, Roditty, and Thorup~\cite{Patrascu12NewInfinity}
obtained a series of oracle with fractional multiplicative stretches.
Baswana Goyal, and Sen~\cite{Baswana09Nearly2ApproximateJournal}
marginally increased the stretch to $(2,3)$ and space to $\Otilde(n^{5/3})$
in order to reduce the preprocessing time to $\Otilde(n^2)$.
The stretch was later reset again to $(2,1)$ by Sommer~\cite{sommer2016all}, 
keeping the improved time complexity.
A successive work by Knudsen~\cite{Knudsen17AdditiveSpanner} removed all additional poly-logarithmic factors in both the construction time and space.

Akav and Roditty~\cite{akav2020almost} proposed, for any $1/2 > \varepsilon > 0$,
an $O(m+n^{2-\Omega(\epsilon)})$-time algorithm that computes a $(2+\epsilon, 5)$ distance oracle with $O(n^{11/6})$ space,
thus breaking the quadratic time barrier for multiplicative stretch below $3$.
Chechik and Zhang~\cite{ChechikZ:22} improved this by offering
both an $(2,3)$-approximate DO with $\Otilde(m + n^{1.987})$ preprocessing time
and an $(2+\varepsilon, c(\varepsilon))$-approximate DO with 
preprocessing time $O(m + n^{\frac{5}{3}-\varepsilon})$,
where $c(\varepsilon)$ is exponential in $1/\varepsilon$.
Both data structures have space $\Otilde(n^{5/3})$ and a constant query time.

Probably closest to our hierarchy in \Cref{thm:hierarchy_DO}
is the distance labeling scheme of Abraham and Gavoille~\cite{Abraham11AffineStretch}.
Seen as a distance oracle for unweighted graphs,
it has a stretch of $(2k{-}2,1)$ for any positive integer $k \ge 2$,
space $\Otilde(n^{1+\frac{2}{2k-1}})$, and query time $O(k)$.
To the best of our knowledge, all previous results for dense graphs 
have a multiplicative stretch of at least $2$.

Goldstein Kopelowitz, Lewenstein, and Porat~\cite{GoldsteinKLP17} showed that, under $k$-Reachability conjecture, any distance oracle with stretch better than $(1+1/t)$ must satisfy a space time trade-off of $S \times T^{1/(t-0.5)}=\widetilde{\Omega}(n^2)$.
This shows that the additive stretch in our oracle presented in \Cref{cor:distance_oracle} 
can be reduced to at most $1$,
even if we are only interested in querying non-neighboring vertex pairs.

It is known that introducing an additive stretch can also help with other parameters
than the multiplicative stretch.
In the Agarwal-Brighten Godfrey result~\cite{AgarwalBrightenGodfrey13DOsStretchLessThan2} above, when introducing an additive stretch of $(2{-}\frac{1}{t})W$ in addition to the $(1+\frac{1}{t})$
multiplicative, the \emph{query time} can be reduced to $O((\mu + \gamma)^{2t-1})$.
Gaur, Sen, and Upadhyay~\cite{baswana2008distance} introduced an additive
stretch to achieve subquadratic \emph{preprocessing time}.
For any integer $k \ge 3$, they gave a $(2k{-}1,2W)$-approximate oracle
with space $O(kn^{1+1/k})$ but a preprocessing time
of $O(\min\{ m + kn^{\frac{3}{2}+\frac{1}{2k}+\frac{1}{2k-2}},\ kmn^{\frac{1}{k}}\})$.\\

\section{Overview}
\label{sec:overview}

It is a common pattern in the design of distance oracles to designate a subset 
(sometimes a hierarchy of sets) of vertices
called \emph{centers}~\cite{ThorupZ05}, \emph{landmark vertices}~\cite{AgarwalBrightenGodfrey13DOsStretchLessThan2} or \emph{pivots}~\cite{Chechik14,Chechik15}.
The oracle stores, for each vertex $v$ in the graph,
the distance from $v$ to all pivots.
Also, $v$ knows its closest pivot $p(v)$.
When given two query vertices $s,t \in V$, the data structure first checks whether $s$ and $t$
are sufficiently ``close''\footnote{%
	The exact definition of ``close'' varies among the different constructions.
} to work out the exact distance $d(s,t)$.
Otherwise, the estimate $d(s,p(s))+d(p(s),t)$ is returned,
which is upper bounded by $d(s,t) + 2 \nwspace d(s,p(s))$.
Since $s$ and $t$ are ``far'' from each other compared to $d(s,p(s))$, 
the estimate has a good stretch.
This observation, of course, is in no way confined to $s$.
For any vertex $v$ on a shortest path from $s$ to $t$,
$d(s,p(v))+d(p(v),t)$ incurs an error of at most $2 \nwspace d(v,p(v))$.
This gives some freedom how much storage space and query time
one is willing to spend
on finding a vertex on the path with a small distance to its closest pivot.  

Our twist to that method is to look at the vicinity of a vertex not in terms of a fixed radius,
but by an absolute bound on the number of considered vertices.
Namely, we define a cut-off value $K$ and store, for each vertex $v$, the $K$ nearest vertices,
regardless of the actual distance to $v$.
Sampling $\Otilde(n/K)$ pivots\footnote{%
	We use the $\Otilde$-notation to hide $\polylog(n)$ factors.
}
ensures that every vertex has a pivot either in or just outside of its $K$-vicinity.
Storing the distances from every vertex to every pivot takes $\Otilde(n^2/K)$ space,
so $K$ is our saving over quadratic space.

One of two things can happen.
If all vertices in the list $K[v]$ around $v$ have a small graph distance to $v$,
there must also be a pivot nearby.
Otherwise, there are elements in $K[v]$ that have a large graph distance,
which we can use to skip ahead in our search for suitable points on the path from $s$ to $t$.
This  set up a win-win strategy.
Consider the graph $H$ in which $v$ has an edge to each member of $K[v]$.
Given a query $(s,t)$ and an approximation parameter $\varepsilon > 0$,
we conduct a bidirectional breath-first search
in $H$ starting from both $s$ and $t$, trimming the search at hop-distance $1/\varepsilon$.
The two searches meeting in one or more points is our definition of $s$ and $t$ being ``close''
(in the above sense).
We can then compute $d(s,t)$ exactly by minimizing $d(s,v)+d(v,t)$ over the intersection points.
Otherwise, we prove that the reason why the searches remained disjoint
was that we could not skip ahead fast enough.
There must have been a vertex on the $s$-$t$-path, touched by one of the searches,
for which all neighbors in $K[v]$ had a small distance to $v$.
This implies that $d(v,p(v))$ is small as well.
Small here means at most  $\frac{\varepsilon}{2} \, d(s,t) + W$,
where $W$ is the weight of the heaviest edge (on a shortest $s$-$t$-path).
As argued above, $d(s,p(v))+ d(p(v),t)$ overestimates
the true distance by at most $2 \nwspace d(v,p(v))$, resulting in an $1+\varepsilon$
multiplicative stretch plus $2W$ additive.

Spacewise, the bottleneck is to store all distances between vertices and pivots.
We device a way to avoid this, further reducing the space,
albeit increasing both the multiplicative and additive stretch.
We are now looking for two vertices $u$ and $v$ on the $s$-$t$-path
that have small distance to their closest pivots $p(u)$ and $p(v)$, respectively.
The portion of the distance between $p(u)$ and $p(v)$ is not stored directly but instead
estimated by another, internal, distance oracle.
Since the inner data structure only needs to answer queries between pivots,
we can get a $2k-1$ multiplicative stretch (for this part) 
with only $\Otilde( (n/K)^{1+1/k})$ space.
This gives a hierarchy of new distance oracles with ever smaller space.

\section{Preliminaries}
\label{sec:prelims}

We let $G = (V,E)$ denote the undirected base graph with $n$ vertices and $m$ edges,
possibly edge-weighted by some function $w \colon E \to [1,W]$.
For any undirected weighted graph $H$, which may differ from the input $G$,
we denote by $V(H)$ and $E(H)$ the set of its vertices and edges, respectively.
Let $P$ be a path in $H$ from a vertex $s \in V(H)$ to $t \in V(H)$, 
its \emph{length} is $|P| = \sum_{e \in E(P)} w(e)$;
in case $H$ is unweighted, we have $|P| = |E(P)|$.
For vertices $u,v \in V(P)$, we let $P[u..v]$ denote the subpath of $P$ from $u$ to $v$.
Let $P = (u_1, \dots, u_i)$ and $Q = (v_1, \dots, v_j)$ be two paths in $H$.
Their \emph{concatenation} is $P \circ Q = (u_1, \dots, u_i, v_1, \dots, v_j)$,
which is well-defined if $u_i = v_1$ or $\{u_i,v_1\} \in E(H)$.
For $s,t \in V(H)$, the \emph{distance} $d_H(s,t)$ 
is the minimum length of any $s$-$t$-path in $H$;
and $d_H(s,t) =+ \infty$ if no such path exists.
We drop the subscripts for the base graph $G$.

Let $\alpha,\beta$ be two non-negative reals with $\alpha \ge 1$.
A \emph{spanner of stretch} $(\alpha,\beta)$, or \emph{$(\alpha,\beta)$-spanner},
is a subgraph $S \subseteq H$ with the same vertex set $V(S) = V(H)$
such that additionally for any two vertices $s,t \in V(S)$, 
it holds that $d_H(s,t) \le d_S(s,t) \le \alpha \cdot d_H(s,t) + \beta$.
We say the stretch is \emph{multiplicative} if $\beta = 0$,
and \emph{additive} if $\alpha = 1$.
The \emph{size} of the spanner is the number of its edges.

A \emph{distance oracle} (DO) for $H$ is a data structure
that reports, upon query $(s,t)$, the distance $d_H(s,t)$.
It has \emph{stretch} $(\alpha,\beta)$ if the reported value $\widehat{d}_H(s,t)$
satisfies $d_H(s,t) \le \widehat{d}_H(s,t) \le \alpha \cdot d_H(s,t) + \beta$.
We measure the space complexity of the oracle in the number of $O(\log n)$-bit machine words.
The size of the graph $H$ does not count against the space, 
unless it is stored explicitly.

\section{Near-Exact Distance Oracles}
\label{sec:DOs}

In this section, we present the construction of our distance oracle in subquadratic space,
trading a small additive stretch for an improved multiplicative one.
For convenience, we restate \Cref{thm:distance_oracle} below.
We first prove the space, query time, and stretch
and then describe how to make the oracle path-reporting.

\distanceoracle*

\noindent
\textbf{Preprocessing and space.}
Let $G$ denote the underlying graph.
First, we make a simplifying assumption.
It is easy to compute the connected components of $G$ in time $O(m)$ (a fortiori in APSP time)
and store in an $O(n)$-sized table the component ID for each vertex.
This allows us to check in constant time whether two vertices have a path between them,
and otherwise correctly answer $d(s,t) = +\infty$.
We can thus assume to only receive queries for vertex pairs in the same component.

\begin{algorithm}[t]
\setstretch{1.33}
\vspace*{.25em}
    $A_{1/\varepsilon}(s) \gets \{v \in V \mid d^{\text{hop}}_H(s,v) \le \lceil 1/\varepsilon \rceil\}$\; 
    $A_{1/\varepsilon}(t) \gets \{v \in V \mid d^{\text{hop}}_H(v,t) \le \lceil 1/\varepsilon \rceil\}$\;  
   $\widehat{d}_1 \gets \min \{d_{1/\varepsilon}(s,v) + d_{1/\varepsilon}(v,t) \mid v \in A_{1/\varepsilon}(s) \cap A_{1/\varepsilon}(t)\}$\; \label{line:d1_minimum}
    $\widehat{d}_2 \gets \min \{d(s,p(v)) + d(p(v),t) \mid v \in A_{1/\varepsilon}(s) \cup A_{1/\varepsilon}(t)\}$\; \label{line:d2_minimum}
    \Return $\min(\widehat{d}_1,\widehat{d}_2)$\; \label{line:global_minimum}
\caption{Query algorithm of the distance oracle in \Cref{thm:distance_oracle} for the query $(s,t)$.
	$d^{\text{hop}}_H$ is the hop-distance in $H$,
	$d_{1/\varepsilon}$ is the minimum length of all
	paths with at most $\lceil 1/\varepsilon \rceil$ hops in $H$,
	and $p(v) \in B$ is the pivot closest to $v$ in $G$.\vspace*{.25em}}
\label{alg:query_DO}
\end{algorithm}

In APSP time, we compute, for every vertex $v \in V$, the list $K[v]$ of its $K$ closest vertices
in $G$ (including $v$ itself) where ties are broken arbitrarily.
Each element of $K[v]$ is annotated with its distance to $v$.
We also sample a set $B$ of vertices, called \emph{pivots}, by including any vertex in $B$
independently with probability $C (\log n)/K$ for a sufficiently large constant $C > 0$.
(Since $K = \omega(\log n)$ this is indeed a probability).
It is easy to show using Chernoff bounds that w.h.p.\
$B$ has $O(\frac{n}{K} \log n)$ elements.
We use $p(v)$ to denote the pivot closest to $v$.
Our data structure stores, for each vertex $v$, the list $K[v]$, the closest pivot $p(v)$,
and the distance $d(v,p(v))$.
Moreover, for each pivot $p \in B$, it additionally stores the distance from $p$ 
to \emph{every} vertex in $G$.
In total, the data structure takes space $O(|B|n + nK) = O(\frac{n^2}{K} \log n + nK)$,
which is $\Otilde(n^2/K)$ for $K = O(\sqrt{n})$.
\vspace*{.5em}

\noindent
\textbf{Query algorithm.}
We use an auxiliary graph $H$ to simplify the presentation of the query algorithm
and the subsequent reasoning.
The graph $H$ has the same vertex set as $G$ and, for each $v \in V$,
an edge from $v$ to any $v' \in K[v]{\setminus}\{v\}$ whose weight is $d(v,v')$.
The \emph{hop-distance} between two vertices $u,v \in V$ in $H$,
is the minimum number of edges on any $u$-$v$-path.

\Cref{alg:query_DO} summarizes how a query $(s,t) \in V^2$ is processed.
Let $A_{1/\varepsilon}(s), A_{1/\varepsilon}(t)$ be the sets of vertices
with hop-distance at most $\lceil 1/\varepsilon \rceil$ from $s$ and $t$, respectively, 
in the graph $H$.
To compute the set $A_{1/\varepsilon}(s)$ we use a slightly modified breath-first search from $s$ (analogously for $A_{1/\varepsilon}(t)$ starting from $t$).
The modification consists in exploring, in each step with current vertex $v$,
\emph{all} neighbors in $K[v]{\setminus}\{v\}$ 
as long as fewer than $\lceil 1/\varepsilon \rceil$ hops have been made.
In particular, the search revisit vertices that have been encountered earlier.
Each time a vertex $v$ is visited,
its estimate of the distance $d_H(s,v)$ is updated to the minimum over all paths explored so far.
In other words, for each $v \in A_{1/\varepsilon}(s)$
the length of the shortest path from $s$ that uses at most $\lceil 1/\varepsilon \rceil$ edges is computed.
We use $d_{1/\varepsilon}(s,v)$ to denote this distance.
Note that $d_{1/\varepsilon}(s,v)$ may overestimate the true distance $d_H(s,v)$ in $H$,
which in turn overestimates $d(s,v)$ in $G$.
Below, we identify some conditions under which the estimate is accurate.
The sets $A_{1/\varepsilon}(s), A_{1/\varepsilon}(t)$ have $O(K^{\lceil 1/\varepsilon \rceil})$ elements and, with the information stored by the distance oracle,
the modified BFSs in $H$ can by emulated in time $O(K^{\lceil 1/\varepsilon \rceil})$.

Recall that $p(v)$ is the pivot closest to $v$
and that the values $d(s,p(v)), d(p(v),t)$ are stored.
The oracle uses $A_{1/\varepsilon}(s), A_{1/\varepsilon}(t)$ to compute
\begin{equation*}
	\widehat{d_1} = \min_{v \in A_{1/\varepsilon}(s) \cap A_{1/\varepsilon}(t)} 
		d_{1/\varepsilon}(s,v) + d_{1/\varepsilon}(v,t)
		\quad\text{and}\ \quad
	\widehat{d_2} = \min_{v \in A_{1/\varepsilon}(s) \cup A_{1/\varepsilon}(t)} d(s,p(v)) + d(p(v),t).
\end{equation*}
It returns the smaller of the two estimates.
The total query time is $O(K^{\lceil 1/\varepsilon \rceil})$.
\vspace*{.5em}

\noindent
\textbf{Stretch.}
For a radius $r$ and vertex $v \in V$, let $B_r(v)$ be the set of all vertices that have distance at most $r$ from $v$.
Define $V_{r} = \{v \in V \mid B_r(v) \subseteq K[v] \nwspace \}$, 
that is, the set of all vertices $v$
that have at most $K$ vertices within distance $r$.
Recall that for any query $(s,t)$ to the distance oracle, 
we can assume that there is a path between $s$ and $t$ in $G$.

\begin{lemma}
\label{lem:property-H}
	Consider two vertices $s,t \in V$ with distance $d(s,t)$ and set 
	$r = \frac{\varepsilon}{2} \nwspace d(s,t)+W$.\footnote{
		The additive term $W$ can be replaced by
		the largest weight $w$ of an edge on $P$.
		This reduces the stretch in \Cref{lem:correctness-weighted} to $(1{+}\varepsilon, 2w)$.
	}
	Let $P$ be a shortest path between $s$ and $t$ in $G$. 
	\begin{enumerate}
		\item[(i)] If all vertices of $P$ lie in $V_r$ the hop-distance between $s$ and $t$
			in $H$ is at most $\lceil 2/\varepsilon \rceil$.
		\item[(ii)] If all vertices of $P$ lie in $V_r$, 
			then $V(P) \cap A_{1/\varepsilon}(s) \cap A_{1/\varepsilon}(t)$
			is non-empty.
		\item[(iii)] If $P$ contains a vertex from $V{\setminus}V_r$, 
		then $(V(P)\setminus V_r) \cap A_{1/\varepsilon}(s)$ or
		 $(V(P)\setminus V_r) \cap A_{1/\varepsilon}(t)$ is non-empty.
	\end{enumerate}
\end{lemma}

\begin{proof}
	Let $d$ abbreviate the distance $d(s,t)$.
	First, consider the case that all vertices of $P$ lie in $V_{r} = V_{\frac{\varepsilon d}{2}+W}$.
	We define a sequence $\sigma_s$ of vertices in $V(P)$,
	namely, a subsequence of $P$ when directed away from $s$.
	Its first vertex is $x_0 = s$;
	each consecutive vertex $x_{i+1}$ is the element of $V(P[x_{i},t]) \cap B_r(x_i)$	
	that maximizes $d(x_i,x_{i+1})$.
	That means $x_{i+1}$ comes after $x_i$ on $P$ when going from $s$ to $t$,
	but it has distance at most $\frac{\varepsilon d}{2}+W$ from $x_i$.
	By the assumption $V(P) \subseteq V_r = \{v \in V \mid B_r(v) \subseteq K[v] \nwspace \}$,
	the sequence $\sigma_s$ corresponds to an actual sequence of hops in the auxiliary graph $H$.
	Moreover, for each $x_i$, 
	$d_{1/\varepsilon}(s,x_i) \le \sum_{j=0}^{i-1} d(x_j,x_{j+1}) = d(s,x_{j+1})$,
	so in this case $d_{1/\varepsilon}$ is exact.
	
	Observe that consecutive vertices in $\sigma_s$ (except possibly the last pair involving $t$)
	have graph distance in $G$ of more than $\varepsilon d/2$.
	Indeed, if we had $d(x_i,x_{i+1}) \le \varepsilon d/2$,
	then the next vertex $v$ on $P$ that comes after $x_i$
	has a higher distance, but still satisfies
	$d(x_i,v) = d(x_i,x_{i+1}) + w(x_{i+1},v) \le \frac{\varepsilon d}{2}+W$.
	This shows that $\sigma_s$ reaches from $s$ to $t$ in at most $\lceil 2/\varepsilon\rceil$ hops.
	
	Symmetrically, we define the sequence $\sigma_t$ by perceiving $P$ as directed away from $t$.
	The first vertex is $y_0 = t$ and $y_{i+1}$ maximizes $d(y_i,y_{i+1})$
	over $V(P[s,y_{i}]) \cap B_r(y_i)$.
	Let $k$ be the smallest index in $\sigma_s$ such that $d(s,x_{k}) \ge d/2$.
	Then, we have $k \le \lceil 1/\varepsilon \rceil$
	by the above observation on the consecutive distances.
	Analogously, let $\ell$ be the minimum index in $\sigma_t$ with $d(t,y_{\ell}) \ge d/2$.
	The predecessor of $y_{\ell}$ thus satisfies $\frac{d}{2} < d(s,y_{\ell-1}) \le \frac{d}{2} + r$
	and therefore $d(x_k,y_{\ell-1}) \le r$.
	This shows that $x_k$ can be reached from $t$ in $H$ via $\ell \le  \lceil 1/\varepsilon \rceil$
	hops by first following $\sigma_t$ until $y_{\ell-1}$ and then hopping to $x_k$.
	In particular, we have $x_k \in V(P) \cap A_{1/\varepsilon}(s) \cap A_{1/\varepsilon}(t)$.
	
	The last part, where the path $P$ contains a vertex from $V{\setminus}V_r$,
	is structurally similar but somewhat simpler.
	Let vertex $z$ be the minimizer of $\min( d(s,z), d(z,t))$ in $V(P){\setminus}V_r$.
	W.l.o.g.\ $z$ is closer to $s$ than to $t$, whence $d(s,z) \le  d/2$.
	We now define the sequence $\sigma_{s}$ as above but let it end in $z$ instead of $s$.
	Note that, by the minimality of $z$, all vertices of $\sigma_{s}$ except for $z$ itself
	are in $V_r$.
	The same argument as above now shows that $\sigma_{s}$ 
	has at most $\lceil 1/\varepsilon \rceil$ vertices, which correspond to hops in $H$.
	That means, $z \in (V(P)\setminus V_r) \cap A_{1/\varepsilon}(s)$.
\end{proof}

To prove the approximation guarantee,
we use the following straightforward application of Chernoff bounds:
with high probability\footnote{%
  An event occurs \emph{with high probability} (w.h.p.) if it has probability at least $1- n^{-c}$ for some constant $c >0$.
  In fact, $c$ can be made arbitrarily large without affecting the asymptotic statements.
} all vertices $v \in V{\setminus}V_r$ satisfy $d(v,p(v)) \le r$.
For if $B_r(v)$ contains more than $K$ elements, it also has a pivot w.h.p.

\begin{lemma}
\label{lem:correctness-weighted}
	For any $s,t \in V$, the distance oracle returns a distance of stretch $(1+\varepsilon,2W)$.
\end{lemma}

\begin{proof}
	The oracle correctly answers $+\infty$ if $s$ and $t$ are in different component.
	Let again $P$ be a shortest $s$-$t$-path, $d = d(s,t)$, and $r = \frac{\varepsilon d}{2} + W$.
	First, note that the returned value is never smaller than $d$ since 
	$\widehat{d_1} = d_{1/\varepsilon}(s,v)+d_{1/\varepsilon}(v,t) \ge d(s,v)+ d(v,t) \ge d(s,t)$
	and $\widehat{d_2} = d(s,p(v)) + d(p(v),t)$ even corresponds to an actual path between $s$ and $t$.
	
	First, assume that all vertices of $P$ are in $V_r$.
	There exists some vertex $v \in V(P) \cap A_{1/\varepsilon}(s) \cap A_{1/\varepsilon}(t)$.
	The returned distance is exact since
	\begin{equation*}
		\widehat{d_1} \le d_{1/\varepsilon}(s,v) + d_{1/\varepsilon}(v,t) = d(s,v) + d(v,t) = d.
	\end{equation*}
	The first equality was argued in \Cref{lem:property-H}, the second one is due to $v$
	being on a shortest $s$-$t$-path.
	Otherwise, there is a vertex $v \in (V(P){\setminus}V_r) \cap (A_{\frac{1}{\varepsilon}}(s) \cup A_{\frac{1}{\varepsilon}}(t))$,
	from which we get w.h.p.\ that
	\begin{equation*}
		\widehat{d}(s,t) \le \widehat{d_2} \le d(s,p(v)) + d(p(v),t)
			\le d(s,v) + d(v,t) + 2 \cdot d(v,p(v)) \le d + 2r
			= (1{+}\varepsilon) \nwspace d + 2W. \qedhere
	\end{equation*} 
\end{proof}

\noindent
\textbf{Reporting paths.}
Only small adaptions are needed to make the oracle path-reporting.
Recall that we use an emulated BFS to compute the distance $d_{1/\varepsilon}(s,v)$
for all $v \in A_{1/\varepsilon}(s)$ (length $d_{1/\varepsilon}(v,t)$
for $v \in A_{1/\varepsilon}(t)$) by iteratively updating the current best length
of a path in $H$ with at most $\lceil 1/\varepsilon \rceil$ edges.
A path through $H$ may not correspond to a path through $G$,
namely, if it uses and edge $\{v,v'\}$ with $v' \in K[v]$
that is not actually present in $G$.
However, any shortest $v$-$v'$-path in $G$ exclusively uses vertices from $K[v]$.

We thus store $K[v]$ in the form of a shortest-path tree in $G$ rooted at $v$
that is truncated after $K$ vertices have been reached.
In each step of the emulated BFS, we explore the neighborhood $K[v]$
in the order given by an actual BFS of the shortest-path tree.
Each time some estimate $d_{1/\varepsilon}(s,v')$ is updated
we also store a pointer to the last vertex from which we reached $v'$.
Furthermore, with each distance $d(v,p)$ for an arbitrary pivot $p \in B$,
we store the first edge on a shortest path from $v$ to $p$.
This does not change the space requirement of the oracle by more than a constant factor.

Suppose the minimum reported in \Cref{line:global_minimum} of \Cref{alg:query_DO}
is attained by $\widehat{d_1}$ and let $v$ be the minimizing vertex in \Cref{line:d1_minimum}.
Following the stored pointers backwards reconstructs a shortest $s$-$v$-path through $A_{1/\varepsilon}(s)$.
Symmetrically, following the pointers forward gives shortest $v$-$t$-path through $A_{1/\varepsilon}(t)$.
Reporting this path in order from $s$ to $t$ can be done by visiting any of the computed edges
at most twice.
If instead $\widehat{d_2}$ is smaller with minimizer $v$ (in \Cref{line:d2_minimum}),
we follow the first edge on a shortest path from $s$ to $p(v)$ (which we have stored)
and likewise for each intermediate vertex we encounter this way between $s$ and $p(v)$.
After $p(v)$, we instead follow along a shortest $t$-$p(v)$-path.

\section{A Hierarchy of Distance Oracles}
\label{sec:hierarchy}

The main bottleneck of the space requirement of the distance oracle in \Cref{thm:distance_oracle}
is to store the distance from every pivot to all vertices of the graph.
We next show how to improve this by instead estimating the distance between pivots with another, internal, distance oracle.
In effect, we trade ever smaller space for higher overall stretch.

\hierarchydo*

In order to prove this, we use the following result from the literature.
Chechik~\cite{Chechik14,Chechik15} gave an improved implementation of the Thorup-Zwick DO~\cite{ThorupZ05}
and also obtained a leaner version of the oracle when restricting the attention to distances between only a subset of the vertices.

\begin{theorem}
\label{thm:lean_ThorupZwick}
    Let $G = (V,E)$ be an undirected weighted graph.
    For any positive integer $k$, and vertex set $B \subseteq V$,
    there is a data structure that reports $(B{\times}B)$-distances with
    multiplicative stretch $2k-1$, space $O(|B|^{1+\frac{1}{k}})$, and constant query time.
    The preprocessing time is $O(n^2 + m \sqrt{n})$.
\end{theorem}

\noindent
\textbf{Preprocessing and space.}
For \Cref{thm:hierarchy_DO}, let the set of pivots $B$ as well as the lists $K[v]$ and closest pivot $p(v)$ be defined as in \Cref{sec:DOs}.
We preprocess the distance oracle of \Cref{thm:lean_ThorupZwick} for pairs of pivots in $B$.
Let $\widehat{D}(p,q)$ denote the estimate of $d(p,q)$ returned by that DO when queried with $p,q \in B$.
Our data structure stores the list $K[v]$ for every $v \in V$ as well as the $(B{\times}B)$-DO.
This takes space $O(nK + (\frac{n}{K} \log n)^{1+1/k}) = O((\frac{n}{K})^{1+1/k} \, \log^{1+1/k} n)$,
assuming $K = O(n^{1/(2k+1)})$.
\vspace*{.5em}

\begin{algorithm}[t]
\setstretch{1.33}
\vspace*{.25em}
    $A_{4k/\varepsilon}(s) \gets \{v \in V \mid d^{\text{hop}}_H(s,v) \le \lceil 4k/\varepsilon \rceil\}$\; 
    $A_{4k/\varepsilon}(t) \gets \{v \in V \mid d^{\text{hop}}_H(v,t) \le \lceil 4k/\varepsilon \rceil\}$\;  
   \eIf{$t \in A_{4k/\varepsilon}(s)$}{\Return $d_{4k/\varepsilon}(s,t)$\;}
   		{\Return $\min \{d(s,p(u)) + \widehat{D}(p(u),p(v)) + d(p(v),t) \mid u \in A_{4k/\varepsilon}(s); \nwspace v \in A_{4k/\varepsilon}(t)\}$\;}
\caption{Query algorithm of the distance oracle in \Cref{thm:hierarchy_DO} for the query $(s,t)$.\\
	$d^{\text{hop}}_H$ is the hop-distance in $H$, $d_{4k/\varepsilon}$ is the minimum length of
	all paths with at most $\lceil 4k/\varepsilon\rceil$ hops in $H$,
	$p(v) \in B$ is the pivot closest to $v$ in $G$,
	and $\widehat{D}$ is the output of the distance oracle in \Cref{thm:lean_ThorupZwick}.\vspace*{.25em}}
\label{alg:query_hierarchy}
\end{algorithm}

\noindent
\textbf{Query algorithm.}
The query algorithm is shown in \Cref{alg:query_hierarchy}.
Recall that we defined the auxiliary graph $H$ by requiring 
that every vertex is connected to its $K$ closest vertices in $G$.
Define $\delta = \varepsilon/(2k)$ and
let $A_{2/\delta}(s), A_{2/\delta}(t)$ the sets of vertices that have a hop-distance 
at most $\lceil 2/\delta \rceil = \lceil 4k/\varepsilon \rceil$ from $s$ and $t$, respectively, in $H$.
If $t$ is found while computing $A_{2/\delta}(s)$, the oracle returns the distance $d_{2/\delta}(s,t)$,
that is, the minimum length of all $s$-$t$-paths in $H$ with at most $\lceil 2/\delta \rceil$ edges.
Otherwise, it reports
\begin{equation*}
	\widehat{d} = \min_{u \in A_{2/\delta}(s),\, v \in A_{2/\delta}(t)} d(s,p(u)) + \widehat{D}(p(u),p(v)) + d(p(v),t)
\end{equation*}
Evaluating that minimum over all pairs $(u,v)$ takes time 
$O(|A_{2/\delta}(s)| \nwspace	 |A_{2/\delta}(t)|) 
	= O( (K^{\lceil 2/\delta\rceil})^2) = O(K^{2\lceil 4k/\varepsilon \rceil})$,
which dominates the query time.
\vspace*{.5em}

\noindent
\textbf{Stretch.}
Fix two query vertices $s,t \in V$ and let $d = d(s,t)$ be their distance.
Recall that the set $V_{\frac{\delta d}{2} + W}$
consists of all those vertices whose ball with radius $\frac{\delta d}{2} + W$ has size at most $K$.

\begin{lemma}
\label{lem:correctness-hierarchy}
	For any $s,t \in V$, the distance oracle returns a distance of stretch $(2k{-}1{+}\varepsilon, 4kW)$.
\end{lemma}

\begin{proof}
	Let $P$ be a shortest $s$-$t$-path in $G$.
	By \Cref{lem:property-H} with $\delta$ in place of $\varepsilon$,
	if all vertices of $P$ are in $V_{\frac{\delta d}{2} + W}$
	then $s$ and $t$ have hop-distance at most $\lceil 2/\delta \rceil$
	in the auxiliary graph $H$, whence $t \in A_{2/\delta}(s)$
	and $d_{2/\delta}(s,t) = d(s,t)$.
	
	Otherwise, the set $V(P){\setminus}V_{\frac{\delta d}{2} + W}$ is non-empty.
    Let $v_s$ be the vertex on $P$
    that is closest to $s$ and does not lie in $V_{\frac{\delta d}{2} + W}$.
    Clearly, we have $v_s \in A_{2/\delta}(s)$.
    For the vertex $v_t \in V(P){\setminus}V_{\frac{\delta d}{2} + W}$ that is
    closest to the other endpoint $t$, we get $v_t \in  A_{2/\delta}(t)$.
    In this case, the output of our oracle is at most
    \begin{align*}
    	\widehat{d} &\le d(s,p(v_s)) + \widehat{D}(p(v_s),p(v_t)) + d(p(v_t),t)
    		 \le d(s,p(v_s)) + (2k{-}1) \nwspace d(p(v_s),p(v_t)) + d(p(v_t),t)\\
    		&\le d(s,v_s) + d(v_s,p(v_s)) + 
    			(2k{-}1) d(p(v_s),v_s) + (2k{-}1) \nwspace d(v_s,v_t) 
    				+ (2k{-}1) \nwspace d(p(v_t),v_t) \, +\\
    		&\quad\  d(p(v_t),v_t) + d(v_t,t)\\
    		&\le (2k{-}1) \nwspace d(s,t) + 2k \, d(v_s,p(v_s)) + 2k \, d(p(v_t),v_t)
    		 \le (2k{-}1) \nwspace d + 4k \left( \frac{\delta d}{2} + W \!\right)\\
    		&= (2k{-}1{+}\varepsilon) \nwspace d + 4kW. \qedhere
    \end{align*}
\end{proof}

\section{Open Problems}
\label{sec:conclusion}

When our distance oracles, namely, \Cref{cor:distance_oracle},
are applied to unweighted graphs they give a stretch $(1{+}\frac{1}{t}, 2)$, 
space $\Otilde(n^{2-\frac{c}{t}})$, and query time $O(n^c)$.
A few natural questions to ask here are.

\begin{itemize}
\item What is the sparsest possible distance oracle for reporting distances of stretch $(1+\epsilon, O(1))$, when allowing an arbitrary sublinear query time?

\item For the special case of $t=1$, the oracle has stretch $(2,2)$, space $\Otilde(n^{2-c})$ and query time $O(n^c)$, where $c\leq 0.5$. 
P\v{a}tra\c{s}cu and Roditty~\cite{PatrascuRoditty14BeyondThorupZwick} showed that the space of any oracle of stretch $2$ in unweighted graphs is lower bounded by $\widetilde \Omega(n^{1.5})$
It remains an open question to obtain a tight bound on query time for $c=0.5$.

\item 
Are there any non-trivial distance oracles with purely additive stretch?
\end{itemize}

For the spectrum of oracles with $O(n^{1.5})$ space, we show it is possible to obtain an oracle of $(2k{-}1{+}\frac{1}{t},4k)$ stretch that has  $\Otilde(n^{1+\frac{1}{k}-\frac{c}{8kt}})$ space, and $O(n^{c})$ query time (\Cref{cor:2k-1+eps-oracle}).
An interesting question to explore in this direction is the following.
For any given $\epsilon,\beta>0$, what is largest $c=c(\epsilon,\beta)$ for which one can construct a distance oracle with stretch $(2k-1+\epsilon,\beta)$, query time $O(n)$, and space $O(n^{1+1/k-c})$?

\bibliographystyle{plainurl}
\bibliography{DSO_bib}

\end{document}